\documentclass[a4paper,UKenglish]{lipics-v2019}
\pdfoutput=1
\usepackage{microtype}
\usepackage{amsthm}
\usepackage{amsmath}
\usepackage{amsfonts}
\usepackage{graphicx}
\usepackage{mathtools}
\usepackage{algorithm}
\usepackage[noend]{algpseudocode}

\usepackage{thmtools}
\usepackage{thm-restate}

\newcommand{\bigo}{\mathcal{O}}
\newcommand{\eps}{\varepsilon}

\EventEditors{Inge Li G{\o}rtz and Oren Weimann}
\EventNoEds{2}
\EventLongTitle{31th Annual Symposium on Combinatorial Pattern Matching (CPM 2020)}
\EventShortTitle{CPM 2020}
\EventAcronym{CPM}
\EventYear{2020}
\EventDate{June 17--19, 2020}
\EventLocation{Copenhagen, Denmark}
\EventLogo{}
\SeriesVolume{161}
\ArticleNo{31}

\begin{document}

\title{Approximating Text-to-Pattern Distance via Dimensionality Reduction}
\author{Przemys\l{}aw~Uzna\'nski}{Institute of Computer Science, University of Wrocław, Poland}{puznanski@cs.uni.wroc.pl}{https://orcid.org/0000-0002-8652-0490}{}

\authorrunning{P. Uznański}
\Copyright{Przemys\l{}aw~Uznański}
\hideLIPIcs
\nolinenumbers


\funding{Supported by Polish National Science Centre grant 2019/33/B/ST6/00298.}

\maketitle
\begin{abstract}
Text-to-pattern distance is a fundamental problem in string matching, where given a pattern of length $m$ and a text of length $n$, over an integer alphabet, we are asked to compute the distance between pattern and the text at every location. The distance function can be e.g. Hamming distance or $\ell_p$ distance for some parameter $p > 0$. Almost all state-of-the-art exact and approximate algorithms developed in the past $\sim 40$ years were using FFT as a black-box. In this work we present $\widetilde\bigo(n/\varepsilon^2)$ time algorithms for $(1\pm\varepsilon)$-approximation of $\ell_2$ distances, and $\widetilde\bigo(n/\varepsilon^3)$ algorithm for approximation of Hamming and $\ell_1$ distances, all without use of FFT. This is independent to the very recent development by Chan et~al. [STOC 2020], where $\bigo(n/\varepsilon^2)$ algorithm for Hamming distances not using FFT was presented -- although their algorithm is much more ``combinatorial'', our techniques apply to other norms than Hamming.
\end{abstract}
\newpage

\section{Introduction}
Text-to-pattern distance is a generalization of a classical pattern matching by incorporating the notion of  similarity (or dissimilarity) between pattern and locations of text. The problem is defined in a following way: for a particular distance function between words (interpreted as vectors), given a pattern of length $m$ and a text of length $n$, we are asked to output distance between the pattern and every $m$-substring of the text.
Taking e.g. distance to be Hamming distance, we are essentially outputting number of mismatches in a classical pattern matching question (that is, not only detecting exact matches, but also counting how far pattern is to from being located in a text, at every position). Such a formulation, for a constant-size alphabet, was first considered  by Fischer and Paterson in \cite{FP:1974}. The algorithm of \cite{FP:1974} uses $\bigo(n \log n)$ time and in substance computes the Boolean convolution of two vectors a constant number of times. This was later extended to $\textrm{poly}(n)$ size alphabets by Abrahamson in~\cite{Abrahamson87,K:1987} with $\bigo(n \sqrt{m \log m})$ run-time.

The lack of progress in Hamming text-to-pattern distance complexity sparked interest in searching for relaxations of the problem, with a hope for reaching linear (or almost linear) run-time. There are essentially two takes on this. First consists of approximation algorithms. Until very recently, the fastest known $(1\pm\eps)$-approximation algorithm for computing the Hamming distances was by Karloff \cite{Karloff93}. The algorithm uses random projections from an arbitrary alphabet to the binary one and Boolean convolution to solve the problem in $\bigo(\eps^{-2} n \log^3 n)$ time. Later Kopelowitz and Porat \cite{KP:15} gave a new approximation algorithm  improving the time complexity to $\bigo (\eps^{-1} n\log^3{n}\log{\eps^{-1}} )$, which was later significantly simplified in Kopelowitz and Porat \cite{KopelowitzP18}, with alternative formulation by Uznański and Studen{\'y} \cite{cpm19a}.

Second widely considered way of relaxing exact text-to-pattern distance is to report exactly only the values not exceeding certain threshold value $k$, the so-called $k$-mismatch problem. 
The very first solution to the $k$-mismatch problem was shown by Landau and Vishkin in \cite{LandauV86} working in $\bigo(nk)$ time, using essentially a very combinatorial approach of taking $\bigo(1)$ time per mismatch per alignment using LCP queries. This initiated a series of improvements to the complexity, with algorithms of complexity $\bigo(n \sqrt{k \log k})$ and $\bigo((k^3 \log k + m)\cdot n/m)$ by Amir et~al. in~\cite{AmirLP04}, later improved to $\bigo((k^2 \log k + m\ \text{poly} \log m)\cdot n/m)$ by Clifford et al. \cite{k-mismatch} and finally $\bigo( (m \log^2 m \log |\Sigma| + k \sqrt{m \log m})\cdot n/m)$ by Gawrychowski and Uznański \cite{GU18} (and following poly-log improvements by Chan et~al. in \cite{approxk}).

Moving beyond counting mismatches, we consider $\ell_1$ distances, where we consider text and pattern over integer alphabet, and distance is sum of position-wise absolute differences. Using techniques similar to Hamming distances, the $\bigo(n \sqrt{m \log m})$ complexity algorithms  were developed by Clifford et~al. and Amir et~al. in \cite{DBLP:conf/cpm/CliffordCI05,Amir2005} for reporting all $\ell_1$ distances. It is a major open problem whether near-linear time algorithm, or even $\bigo(n^{3/2-\varepsilon})$ time algorithms, are possible for such problems. A conditional lower bound was shown by Clifford in \cite{Clifford}, via a reduction from matrix multiplication. This means that existence of combinatorial algorithm with $\bigo(n^{3/2-\varepsilon})$ run-time solving the problem for Hamming distances implies  combinatorial algorithms for Boolean matrix multiplication with $\bigo(n^{3-\delta}) $ run-time, which existence is unlikely. Looking for unconditional bounds, we can state this as a lower-bound of $\Omega(n^{\omega/2})$ for Hamming distances pattern matching, where $2 \le \omega < 2.373$ is the matrix multiplication exponent. Later, complexity of pattern matching under Hamming distance and under $\ell_1$ distance was proven to be identical (up to poly-logarithmic terms), see Labib et~al. and Lipsky et~al. \cite{GLU:2018,DBLP:journals/ipl/LipskyP08a}.

Once again, existence of such lower-bound spurs interest in approximation algorithm for $\ell_1$ distances. Lipsky and Porat \cite{DBLP:journals/algorithmica/LipskyP11} gave a deterministic algorithm with a run time of  $\bigo(\frac{n}{\varepsilon^2} \log m \log U)$, while later Gawrychowski and Uznański \cite{GU18} have improved the complexity to a (randomized) $\bigo(\frac{n}{\varepsilon} \log^2 n \log m \log U)$, where $U$ is the maximal integer value on the input. Later \cite{cpm19a} has shown that such complexity is in fact achievable (up to poly-log factors) with a deterministic solution.

Considering other norms, we mention following results. First, that for any $p>0$ there is $\ell_p$ distance $(1\pm\varepsilon)$-approximated algorithm running in $\widetilde{\bigo}(n/\varepsilon)$ time by \cite{cpm19a}. More importantly, for specific case of $p=2$ (or more generally, constant, positive even integer values of $p$) the exact problem reduces to computation of convolution, as observed by~\cite{DBLP:journals/algorithmica/LipskyP11}.

\paragraph*{Text-to-pattern distance via convolution}  Consider the case of computing $\ell_2$ distances. We are computing output array $O[1\ ..\ n-m+1]$ such that $O[i] = \sum_j (T[i+j] - P[j])^2.$ However, this is equivalent to computing, for every $i$ simultaneously, the value of $\sum_{j} T[i+j]^2 + \sum_j P[j]^2 - 2 \sum_j T[i+j] P[j]$. While the terms $\sum_{j} T[i+j]^2$ and $ \sum_j P[j]^2$ can be easily precomputed in $\bigo(n)$ time, we observe (following \cite{DBLP:journals/algorithmica/LipskyP11}) that $\sum_j T[i+j] P[j]$ is essentially a convolution. Indeed, let $P^R$ denote reverse string to $P$. Then
$$\sum_j T[i+j] P[j] = \sum_j T[i+j] P^R[m+1-j] = \sum_{j+k = m+1+i} T[j] P'[k] = (T \circ P^R)[m+1+i].$$
Since $T \circ P^R$ can be computed efficiently this provides a very strong tool in constructing text-to-pattern distance algorithms. Almost all of the discussed results use convolution as a black-box. For example, by appropriate binary encoding we can compute using a single convolution the number of Hamming mismatches generated by a single letter $c \in \Sigma$, which is a crucial observation leading to computation of exact Hamming distances in $\bigo(n \sqrt{n \log n})$ time. Other results rely on projecting large alphabets into smaller ones, e.g.  \cite{Karloff93,KopelowitzP18,cpm19a}.

Convolution over integers is computed by FFT in $\bigo(n \log n)$ time. This requires actual embedding of integers into field, e.g. $\mathbb{F}_p$ or $\mathbb{C}$. This comes at a cost, if e.g. we were to consider text-to-pattern distance over (non-integer) alphabets that admit only field operations, e.g. matrices or geometric points.
Convolution can be computed using ``simpler'' set of operations, that is just with ring operations in e.g. $\mathbb{Z}_p$ using Toom-Cook multiplication \cite{toom1963complexity}, which is a generalization of famous divide-and-conquer Karatsuba's algorithm \cite{karatsuba1963multiplication}. This however comes at a cost, with Toom-Cook algorithm taking $\bigo(n 2^{\sqrt{2 \log n}} \log n)$ time, and increased complexity of the algorithm. 

Computing convolution comes with another string attached -- it is inefficient to compute/sketch in the streaming setting. All of the efficient streaming text-to-pattern distance algorithms \cite{approxk,k-mismatch,CKP19,HDstream,Porat:09,DBLP:conf/icalp/GolanKP18,DBLP:journals/corr/abs-1907-04405} use some form of sketching and are actually avoiding convolution computation. The reason for this is that convolution does not admit efficient sketching schemes other than with additive error, that is any algorithm based on convolution is supposed to make the same error of estimation in small and large distance regime.

\paragraph*{Our results}
We present approximation algorithm for computing the $\ell_2$ text-to-pattern distance in $\widetilde\bigo(n/\varepsilon^2)$ time, where $\widetilde\bigo$ hides $\text{poly}\log n$ terms. Our algorithm is convolution-avoiding, and in fact it uses mostly additions and subtractions in its core part (some non-ring operations are necessary for output-scaling and hashing). We thus claim our algorithm to be more ``combinatorial'', in the sense that it does not rely on field embedding and FFT computation. Our algorithm is also first non-trivial algorithm for text-to-pattern distance computation with other norms (than Hamming, which was presented recently in \cite{approxk}).

\begin{restatable}{theorem}{mainldwa}
\label{th:mainl2}
Text-to-pattern $\ell_2$ distances  can be approximated by an algorithm using only basic arithmetic operations and not using convolution. The approximation is $1\pm\varepsilon$ multiplicative with high probability, computed in $\bigo(\frac{n \log^3 n}{\varepsilon^2})$ time.
\end{restatable}

This mirrors the recent development of \cite{approxk} where a combinatorial algorithm for Hamming distances was presented with $\bigo(n/\varepsilon^2)$ run-time. However, our techniques are general enough so that we can construct algorithm for $\ell_1$ norm (and Hamming), however with $\widetilde\bigo(n/\varepsilon^3)$ run-time.

\begin{restatable}{theorem}{mainham}
\label{th:mainHam}
Text-to-pattern Hamming distances  can be approximated by an algorithm using only basic arithmetic operations and not using convolution. The approximation is $1\pm\varepsilon$ multiplicative  with high probability, computed in $\bigo(\frac{n \log^4 n}{\varepsilon^3})$ time.
\end{restatable}

\begin{restatable}{theorem}{mainljeden}
\label{th:mainL1}
Text-to-pattern $\ell_1$ distances over alphabet $[u]$ for some constant $u = \textrm{poly}(n)$ can be approximated by an algorithm using only basic arithmetic operations and not using convolution. The approximation is $1\pm\varepsilon$ multiplicative  with high probability, computed in $\bigo(\frac{n \log^2 n (\log^2 n + \log^4 u)}{\varepsilon^3})$ time.
\end{restatable}

We present two novel techniques, to our knowledge never used previously in this setting. First, we show that a ``mild'' dimensionality reduction (linear map reducing from dimension $2d$ to $d$, while preserving $\ell_2$ norm) can be used to repeatedly compress word, and produce sketches for its every $m$-subword. Second, we show an approximate embedding of $\ell_1$ space into $\ell_2^2$, that can be efficiently computed.
We believe our techniques are of independent interest, both to stringology and general algorithmic communities.

\section{Definitions and preliminaries.}

\paragraph*{Distance between strings.}
Let $X = x_1 x_2 \ldots x_n$ and $Y = y_1 y_2 \ldots y_n$ be two strings. 
We define their $\ell_2$ distance as
$$\|X-Y\| = \left(\sum_i |x_i - y_i|^2 \right)^{1/2}.$$
More generally, for any $p > 0$, we define their $\ell_p$ distance as
$$\|X-Y\|_p = \left(\sum_i |x_i - y_i|^p \right)^{1/p}.$$
Particularly, the $\ell_1$ distance is known as the \emph{Manhattan distance}. By a slight abuse of notation, we define the $\ell_0$ (Hamming distance) to be
$$\|X-Y\|_0 = \sum_i |x_i-y_i|^0 = | \{i : x_i \not= y_i\}|,$$
where $x^0 = 1$ when $x \not= 0$ and $0^0 = 0$.

\paragraph*{Text-to-pattern distance.}
For text $T = t_1t_2\ldots t_n$ and pattern $P = p_1p_2\ldots p_m$, the text-to-pattern $d$-distance is defined as an array $S_d$ such that, for every $i$, $S_d[i] = d(T[i+1\ ..\ i+m],P)$. Thus, for $\ell_p$ distance $S_{\ell_p}[i] = \left(\sum_{j=1}^m |t_{i+j}-p_j|^p\right)^{1/p}$, while for Hamming distance $S_{\textrm{HAM}}[i] = |\{ j : t_{i+j} \not= p_j \}|$. Then $(1 \pm \varepsilon)$-approximated distance is defined as an array $S_{\varepsilon}$ such that, for every $i$, $(1-\varepsilon) \cdot S_d[i] \le S_{\varepsilon}[i] \le (1+\varepsilon) \cdot S_d[i]$.


\section{Sketching via dimensionality reduction}

Sketching is a tool in algorithm design, where a large object is summarized succinctly, so that some particular property is approximately preserved and some predefined operations/queries are still supported. Our interest lies on sketches that preserve $\ell_2$ distances, for which we use the standard tools from dimensionality reduction.

\begin{theorem}[Johnson-Lindenstrauss \cite{JLsketches}]
Let $P \subseteq \mathbb{R}^m$ be of size $m$. Then for some $d = \bigo(\frac{\log m}{\varepsilon^2})$ there is linear map $A \in \mathbb{R}^{d \times m}$ such that 
$$\forall_{x,y \in P} \|Ax-Ay\| = (1 \pm \varepsilon) \|x - y\|.$$
\end{theorem}

A map that preserves $\ell_2$ distances is useful. Our goal is to construct a linear map such that we can apply the map to $P$ and to every $m$-substring of $T$ simultaneously and computationally efficiently. For this, we need to actually use constructive version of Johnson-Lindenstrauss lemma.
\begin{theorem}[Achlioptas \cite{DBLP:journals/jcss/Achlioptas03}]
\label{th:jlachlioptas}
Consider a probability distribution $\mathcal{D}$ over matrices $\mathbb{R}^{m \times d}$ defined as follow so that each matrix entry is either $-1$ or $1$ independently and uniformly at random. Then for any $x \in \mathbb{R}^m$ there is $$\Pr_{A \sim \mathcal{D}} \left(\frac{1}{\sqrt{d}} \|A x\| = (1 \pm \varepsilon) \|x\|\right) \ge 1 - \delta$$ if only $d = \bigo(\frac{\log \delta^{-1}}{\varepsilon^2})$ is large enough.
\end{theorem}
Computing such dimension-reduction naively takes $\bigo(md)$ time. However better constructions are possible.
\begin{theorem}[Sparse JL, c.f. \cite{DBLP:conf/soda/CohenJN18,DBLP:journals/jacm/KaneN14}]
There is probability distribution $\mathcal{S}$ over matrices of dimension $d \times m$ with elements from $\{-1,0,1\}$, for large enough $d = \bigo(\frac{\log \delta^{-1}}{\varepsilon^2})$, such that each column has only $s = \bigo(d \varepsilon)$ non-zero elements and for any vector $x \in \mathbb{R}^m$ there is 
$$\Pr_{A \sim \mathcal{S}} \left(\frac{1}{\sqrt{s}} \|A x\| = (1 \pm \varepsilon) \|x\|\right) \ge 1 - \delta.$$
\end{theorem}
Such matrices can be easily drawn from the distribution by selecting the $s$ positions in each column independently at random and then filling them uniformly at random with $\{-1,1\}$. The advantage of this is that single dimensionality reduction operation is computed in $\bigo(s m)$ time which is $\varepsilon^{-1}$ factor faster than for dense matrices.

We now state the take-away from this section, which is our main technical tool to be used in the following.
\begin{corollary}
\label{cor:dimred}
For $d = \bigo(\frac{\log n}{\varepsilon^2})$ large enough there is a probability distribution $\mathcal{F}$ of linear maps $\varphi : \mathbb{R}^d \times \mathbb{R}^d \to \mathbb{R}^d$ such that: 
\begin{enumerate}
\item $\varphi(x,y) = A_0x + A_1y$ can be evaluated in $\bigo(d^2 \varepsilon) = \bigo(\frac{\log^2 n}{\varepsilon^3})$ time, 
\item $\Pr_{\varphi \sim \mathcal{F}} \left( \| \varphi(x,y) \|^2 = (1\pm\varepsilon)( \|x\|^2 + \|y\|^2) \right) = 1-n^{-\Omega(1)},$
\item both $A_0$ and $A_1$ are $\{-\frac{1}{\sqrt{s}},0,\frac{1}{\sqrt{s}}\}$-matrices  where $s = \bigo(d \varepsilon)$ is the sparsity of each column of $A_0$ and $A_1$.
\end{enumerate}
\end{corollary}

\section{Algorithm for $\ell_2$ distances.}
We first use Corollary \ref{cor:dimred} to construct dimensionality reduction with guarantees similar to Johnson-Lindenstrauss (reducing dimension $n$ to dimension $\widetilde\bigo(\varepsilon^{-2})$). In the following we assume that $d = \bigo(\frac{\log n}{\varepsilon^2})$ is large enough. 
We show a procedure which assumes that $m$ is divisible by $d$, and denote $s = \frac{m}{d}$. We assume $s$ is a power of two, and if the case is otherwise, we can always pad input with enough zeroes at the end (we can do this, since extra zeroes have no effect on the output of linear map). We also denote $k = \log_2 s$.
\begin{algorithm}[h!!!]
\label{alg:alg1}
\begin{algorithmic}[1]
\State Input: $x \in \mathbb{R}^m$.
\State Output: $v \in \mathbb{R}^d$.
\Procedure{SingleSketch}{$x$}
\State Pick  $k$ fully independent maps $\varphi_1,\ldots,\varphi_k$ as in Corollary \ref{cor:dimred}.
\State Partition input $x = (x_1,\ldots,x_m)$ into $s$ vectors $v^{(0)}_1,\ldots,v^{(0)}_s$ where $v^{(0)}_i \gets (x_{d\cdot(i-1)+1},\ldots,x_{d\cdot i})$.
\For{$i \gets 1\ ..\ k$}
\For{$j \gets 1\ ..\ 2^{k-i}$}
\State $v^{(i)}_j \gets \varphi_i(v^{(i-1)}_{2j-1},v^{(i-1)}_{2j})$
\EndFor
\EndFor
\State \Return $v = v^{(k)}_1$.
\EndProcedure
\end{algorithmic}
\caption{At each level $i$, we partition its vectors into $2^{k-i}$ pairs, and compress each pair using $\varphi_i$ producing vectors for level $i+1$.}
\end{algorithm}
\newline
We then have the following
\begin{theorem}
\label{th:dimred}
Given input $x \in \mathbb{R}^m$, and $\varepsilon \le \frac{1}{k}$, procedure \textsc{SingleSketch} outputs $v \in \mathbb{R}^d$ such that 
$$\|v\| = (1 \pm \bigo(k \varepsilon)) \|x\|$$
with high probability, in time  $\bigo(\frac{m \log n}{\varepsilon})$.
The map $x \to v$ is linear.
\end{theorem}
\begin{proof}
We first bound the stretch. Denote by
$$\alpha_i = \sum_{j} \|v^{(i)}_j\|^2.$$
Naturally, 
\begin{align*}
\alpha_0 &= \sum_{j} \|v^{(0)}_j\|^2 = \sum_{j=1}^s (x_{d \cdot(j-1)+1}^2 + \ldots + x_{d \cdot j}^2) = \sum_{j=1}^n x_j^2 = \|x\|^2.
\end{align*}
Moreover, by Corollary~\ref{cor:dimred}
\begin{align*}
\alpha_i &= \sum_{j=1}^{2^{k-i}} \|v^{(i)}_j\|^2 = \sum_{j=1}^{2^{k-i}} (1 \pm \varepsilon) ( \|v^{(i-1)}_{2j-1}\|^2 +  \|v^{(i-1)}_{2j}\|^2)\\
 &= (1 \pm \varepsilon) \sum_{j=1}^{2^{k-i+1}} \|v^{(i-1)}_{j}\|^2 = (1 \pm \varepsilon) \alpha_{i-1}
\end{align*}
We could apply Corollary~\ref{cor:dimred} at this step since for any usage of map $\varphi_i$, its inputs are independent from actual choice of $\varphi_i$ (e.g. are result of processing $x$ and $\varphi_1, \ldots, \varphi_{i-1}$).  Then we have $\|v\|^2 = \alpha_k = (1\pm\varepsilon)^k \alpha_0 = (1 \pm \varepsilon)^k \|x\|^2$. Since $\varepsilon \le \frac{1}{k}$, the claimed bound follows.

We then observe that the map is linear, since every building step of the map is linear. The total number of times we apply one of $\varphi_1,\ldots,\varphi_k$ is $\bigo(m/d)$, so the total run-time is $\bigo(\frac{m}{d} d^2 \varepsilon)$.
\end{proof}

We then extend the algorithm to a scenario where for an input word (vector) $x \in \mathbb{R}^n$ we compute the same dimensionality reduction for all $m$-subwords of $x$ that start at all the positions divisible by $d$. In the following we assume that $d$ divides $n$, and denote $t = \frac{n-m}{d}+1$ to be the number of such $m$-subwords. If its not the case, input can be padded with enough zeroes at the end.

\begin{algorithm}[h!!!]
\begin{algorithmic}[1]
\State Input: $x \in \mathbb{R}^n$.
\State Output: $v_1,\ldots,v_t \in \mathbb{R}^d$ for $t = \frac{n-m}{d}+1.$
\Procedure{AllSketch}{$x$}
\State Let $\varphi_1,\ldots,\varphi_k$ be $k$ fully independent maps used in procedure \textsc{SingleSketch}.
\State Partition input $x = (x_1,\ldots,x_n)$ into $n/d$ vectors $v^{(0)}_1,\ldots,v^{(0)}_{\frac{n}{d}}$ where~$v^{(0)}_i~\gets~(x_{d\cdot(i-1)+1},\ldots,x_{d\cdot i})$.
\For{$i \gets 1\ ..\ k$}
\For{$j \gets 1\ ..\ (\frac{n}{d}-2^i+1)$}
\State $v^{(i)}_j \gets \varphi_i(v^{(i-1)}_{j},v^{(i-1)}_{j+2^{i-1}})$
\EndFor
\EndFor
\State \Return $v^{(k)}_1, \ldots, v^{(k)}_t$.
\EndProcedure
\end{algorithmic}
\caption{}
\end{algorithm}

\begin{theorem}
\label{th:dimredsubstrings}
Given input $x \in \mathbb{R}^n$, denote by $y_1,\ldots,y_t \in \mathbb{R}^m$ vectors such that $y_i = (x_{1+(i-1)d}, \ldots, x_{m+(i-1)d})$. For $\varepsilon \le \frac{1}{k}$ procedure $\textsc{AllSketch}$ outputs $v_1,\ldots,v_t \in \mathbb{R}^d$ such that 
$$\|v_j\| = (1 \pm \bigo(k \varepsilon)) \|y_j\|$$
with high probability, in time  $\bigo(\frac{n \log^2 n}{\varepsilon})$. Moreover, the map $y_i \to v_i$ is linear and identical to map from Theorem~\ref{th:dimred}.
\end{theorem}
\begin{proof}
The proof follows from inductive observation that $\|v_{j}^{(i)}\|^2 = (1 \pm \varepsilon)^i(\|v^{(0)}_{j}\|^2 + \ldots \|v^{(0)}_{j+2^i-1}\|^2)$, which results in 
\begin{align*}
\|v_j\|^2 &= (1 \pm \varepsilon)^k \sum_{i=1}^s \|v^{(0)}_{j+i}\|^2\\
&=(1\pm\varepsilon)^k \sum_{i=1}^m \|x_{i+(j-1)d}\|^2\\
&= (1\pm\varepsilon)^k \|y_j\|^2.
\end{align*}
The rest of the proof follows reasoning from Theorem~\ref{th:dimred}.
\end{proof}

\mainldwa*
\begin{proof}
First, we note that for simplicity we compute  $(\ell_2)^2$ distances since they are additive when taken under concatenation of inputs (unlike $\ell_2$), that is $\|x \circ y - u \circ v\|^2 = \|x - u\|^2 + \|y-v\|^2$ for equal length $x,u$ and equal length $y,v$. 

We then assume w.l.o.g. that $n$ is divisible by $d$. We then observe that contribution of any fragment of pattern to distance at every text location can be computed naively in $\bigo(c \cdot n)$ time where $c$ is fragment length. We are thus safe to discard any suffix of pattern of length $\bigo(d)$ as this time is absorbed in total computation time. So we fix $h = \bigo(\log n/\varepsilon)$ and assume w.l.o.g. that $m' = m-2h$ is divisible by $d$. 

We denote by $\varepsilon' = \Omega(\varepsilon/\log n)$ such value that guarantees $(1\pm\varepsilon)$-approximation in Theorem~\ref{th:dimred} and Theorem~\ref{th:dimredsubstrings}.
First, assume for simplicity that $\frac{m'}{d}$ is a power of two. We then consider $P_0,\ldots, P_h$, the $(h+1)$ distinct $m'$-substrings of $P$, and for each we run procedure \textsc{SingleSketch} on each of them, so by Theorem~\ref{th:dimred} we compute their sketches in total $\bigo(\frac{m \log n}{\varepsilon'} h)$ time. Similarly, for text $T$ we run \textsc{AllSketch} $\frac{d}{h}$ times to compute sketches of all $m'$-substrings of $T$ starting at positions $1, h+1, 2h+1,\ldots$. By Theorem~\ref{th:dimredsubstrings} this takes $\bigo(\frac{n \log^2 n}{\varepsilon'} \cdot \frac{d}{h})$ time. Both steps take thus $\bigo(\frac{n \log^3 n}{\varepsilon^2})$ time, and maps used to compute sketches in both steps are linear.

We now observe that for any starting position $t$, the substring $T[t\ ..\ (t+m'-1)]$ can be partitioned into $T_1 = T[t\ ..\ t_1]$, $T_2 = T[t_1+1\ ..\ t_2]$ and $T_3 = T[t_2+1\ ..\ (t+m'-1)]$, where length of $T_1$ and $T_3$ is at most $2h$, length of $T_3$ is $m'$ and $t_1$ and $t_2$ are multiplies of $h$. We then compute the distances between corresponding fragments of $T$ and $P$ as follows (where we consider corresponding partitioning of $P$ into $P_1$, $P_2$ and $P_3$): computing $\|T_1 - P_1\|^2$ and $\|T_3 - P_3\|^2$ takes $\bigo(h)$ each ($\bigo(nh)$ in total for all alignments), while $(1\pm\varepsilon)$ approximating $\|T_2 - P_2\|^2$ follows from pre-computed sketches.

We now discuss the general case when $\frac{m'}{d}$ is not a power of two. However we then observe that $m'$ can be represented as $m' = d (2^{i_1} + \ldots + 2^{i_s})$ where $s \le \log n$. And so the necessary computation require actually querying $s$ different sketches for fragments of length $d \cdot 2^{i_\ell}$. To avoid unnecessary $\bigo(\log n)$ overhead in time (and repeating running the preprocessing steps $\log n$ times for many various lengths of fragments) we observe that all the necessary sketches are already computed as temporary values in procedures \textsc{SingleSketch} and \textsc{AllSketch}.
\end{proof}

\section{Hamming and $\ell_1$ distances.}
We now briefly discuss how to use our framework for approximating other norms. We first recall the classical result by \cite{Karloff93}.
\begin{lemma}[\cite{Karloff93}]
\label{ref:hamembed}
Let $d = \bigo(\log n/\varepsilon^2)$ be large enough. Consider $\mu : \Sigma \to \{0,1\}^d$ where each $\varphi(c)$ is chosen uniformly and independently at random. Then 
$$\forall_{c_1 \not= c_2} \| \mu(c_1) - \mu(c_2) \|^2 = (1 \pm \varepsilon) \cdot \frac{d}{2}$$
with high probability.
\end{lemma}
We note that we assumed that the dimension $\bigo(\log n/\varepsilon^2)$ of map $\mu$ matches value of $d = \bigo(\log n/\varepsilon^2)$ from dimensionality-reductions in previous section. This can be easily ensured w.l.o.g.  as we can always either pad with extra zeroes each image of $\mu$ mapping, or add extra null coordinates to dimensionality reduction.
Extending the mapping from letters to words, that is for $w=c_1\ldots c_k \in \Sigma^*$ denote $\mu(w) = \mu(c_1) \ldots \mu(c_k)$, we have a corollary:
\begin{corollary}
\label{ref:embeddingham}
For $\mu$ as in Lemma~\ref{ref:hamembed}, and any two words $u,v \in \Sigma^n$, there is
$$\|\mu(u) - \mu(v)\|^2 = (1 \pm \varepsilon) \cdot \frac{d}{2} \|u - v\|_0$$
with high probability.
\end{corollary}
This allows us to estimate Hamming distance between words from $\ell_2^2$ distance between the respective embeddings, which are of length $\bigo(\frac{n \log n }{\varepsilon^2})$.

\mainham*
\begin{proof}
By Corollary~\ref{ref:embeddingham} it is enough to estimate the $\ell_2^2$ text-to-pattern distance between embedded words $\mu(P)$ and $\mu(T)$ at starting positions $1,d+1,2d+1,\ldots$. We use procedure $\textsc{SingleSketch}$ to compute sketch of $\mu(P)$, and procedure $\textsc{AllSketch}$ to compute sketch of every $(dm)$-substring of $\varphi(T)$ starting at positions $1, d+1, 2d+1,\ldots$. Former takes $\bigo(\frac{n \log^2 n}{\varepsilon^2 \varepsilon'})$ time, and latter takes $\bigo(\frac{n \log^3 n}{\varepsilon^2 \varepsilon'})$ time, where we set $\varepsilon' = \Omega(\varepsilon/k)$ so that error from sketching accumulates to $1 \pm \bigo(\varepsilon)$ in total. All in all this gives $\bigo(\frac{n \log^4 n}{\varepsilon^3})$ time algorithm.
\end{proof}

We now proceed to $\ell_1$ distances. Our goal is to construct a mapping $f : [u] \to \{0,1\}^d$ that embeds $\ell_1$ into $\ell_2^2$ approximately. That is, we require $\forall_{a,b \in [u]} |a-b| \sim (1 \pm \varepsilon) \|f(a) - f(b)\|^2$ where $\sim$ hides constant factors. The existence of such map can be easily shown:
(i) Take exact map $f_1 : [u] \to \{0,1\}^u$ defined as $f_1(a) = 1^a0^{u-a}$, (ii) Take any $\ell_2$ dimensionality-reduction map $f_2 : \{0,1\}^u \to \{0,1\}^d$, (iii) set $f = f_2 \circ f_1$.
However, our goal is to compute such $f$ faster than in time proportional to universe size $u$. We do it by running first a preprocessing phase, and then a fast computation procedure.

\begin{algorithm}[h!!!]
\begin{algorithmic}[1]
\Procedure{Preprocess}{$u$}
\State Pick  $\log (u/d)$ fully independent maps $\varphi'_1,\ldots,\varphi'_{\log (u/d)}$ as in Corollary \ref{cor:dimred}.
\State $s_0 \gets (1,1,\ldots,1) \in \mathbb{R}^d$.
\For{$i \gets 1\ ..\ \log(u/d)$}
\State $s_i \gets \varphi'_i(s_{i-1},s_{i-1})$
\EndFor
\EndProcedure
\Procedure{Project}{$x \in [u]$, $c$}
\If{$c = 0$}
\State \Return $(\underbrace{1,1,\ldots,1}_{x},\underbrace{0,\ldots,0}_{ d-x})$
\ElsIf{$x < \frac{1}{2} d \cdot 2^c$}
\State \Return $\varphi'_c( \textsc{Project}(x, c-1), (0,\ldots,0) )$
\Else
\State \Return $\varphi'_c( s_{c-1}, \textsc{Project}(x-\frac{1}{2} d \cdot 2^c, c-1))$
\EndIf
\EndProcedure
\end{algorithmic}
\caption{}
\end{algorithm}
\begin{lemma}
\label{lem:l1dimred}
$\psi: x \to \textsc{Project}(x, \log(u/d))$ represents a linear map $[u] \to \mathbb{R}^d$ that embeds approximately $\ell_1$ to $\ell_2^2$, that is
$$|x - y| = (1 \pm \bigo(\varepsilon \log u)) \|\psi(x) - \psi(y)\|^2$$
with high probability. Moreover, $\psi$ takes $\bigo(\frac{\log^2 n \log u}{\varepsilon^3})$ time to evaluate.
\end{lemma}
\begin{proof}
Let us define informally $\pi_i = \varphi'_i(\varphi'_{i-1}(\ldots, \ldots),\varphi'_{i-1}(\ldots, \ldots))$ to be unfolded version of  $\varphi'$, that is a linear map $\mathbb{R}^{d \cdot 2^i} \to \mathbb{R}^d$. Formally $\pi_0 = \textrm{id}$, and for $x = (x_1,\ldots,x_{d \cdot 2^i})$, defining 
$$\pi_i( (x_1,\ldots,x_{d \cdot 2^i}) ) = \varphi'_i( \pi_{i-1}(x_{\textrm{left}}), \pi_{i-1}(x_{\textrm{right}})),$$
where  $x_{\textrm{left}} = (x_1, \ldots, x_{d \cdot 2^{i-1}})$, $x_{\textrm{right}} =  (x_{d \cdot 2^{i-1}+1}, \ldots, x_{d \cdot 2^{i}})$.

We now observe that $s_i = \pi_i( (\underbrace{1,\ldots,1}_{2^i d}) )$ and then (by induction) 
$$\textsc{Project}(x,i) = \pi_i( (\underbrace{1,1,\ldots,1}_{x},\underbrace{0,\ldots,0}_{2^id-x}) ).$$
Inductively, each iteration $1,..,\log(u/d)$ results in extra multiplicative $(1\pm\varepsilon)$ distortion. 
Computation time is dominated by applications of $\varphi'_1,\ldots,\varphi'_{\log(u/d)}$, both in the preprocessing time and the evaluation time. Since each linear map $\varphi'_i$ is applied in time $\bigo(\frac{\log^2 n}{\varepsilon^3})$, the time complexity bound follows.

\end{proof}

\mainljeden*
\begin{proof}
We use Lemma~\ref{lem:l1dimred} to reduce the problem to estimating $\ell_2^2$ text-to-pattern distance between $\psi(P)$ and $\psi(T)$ at starting positions $1,d+1,2d+1,\ldots$. We use procedure $\textsc{SingleSketch}$ to compute sketch of $\mu(P)$, and procedure $\textsc{AllSketch}$ to compute sketch of every $(dm)$-substring of $\varphi(T)$ starting at selected positions. Denote by $\varepsilon' = \Omega(\varepsilon/k)$ the stretch constant in procedures $\textsc{SingleSketch}$ and $\textsc{AllSketch}$, and by $\varepsilon'' = \Omega(\varepsilon/\log u)$ the stretch constant in procedures \textsc{Project} and \textsc{Preprocess}. The total run-time of \textsc{AllSketch} is then $\bigo(\frac{n \log^3n}{\varepsilon^2 \varepsilon'}) = \bigo(\frac{n\log^4n}{\varepsilon^3})$ and total run-time of computing $\psi(T)$ and $\psi(P)$ is $\bigo(\frac{n \log^2 n \log u}{(\varepsilon'')^3}) = \bigo(\frac{n \log^2 n \log^4 u}{\varepsilon^3})$.
\end{proof}

\bibliographystyle{splncs04}
\bibliography{bib}

\end{document}